\documentclass{revtex4}

\usepackage{graphicx}
\usepackage{amsmath,amssymb,amsthm,amsfonts,amstext}
\usepackage{bm}

\newtheorem{theorem}{Theorem}

\newtheorem{proposition}[theorem]{Proposition}

\newtheorem{corollary}[theorem]{Corollary}
\newtheorem{lemma}[theorem]{Lemma}

\newcommand{\Ope}[1]{\mathcal{#1}}
\newcommand{\Bod}[1]{\mathbb{#1}}
\newcommand{\Spa}[1]{\mathbf{#1}}

\usepackage[usenames]{color}
\begin{document}

\title{Variational Necessary and Sufficient Stability Conditions for Inviscid Shear Flow}

\author{M. Hirota$^1$, P. J. Morrison$^2$, Y. Hattori$^1$}

\address{$^1$ Tohoku University, Sendai, Miyagi 980-8677, Japan\\
$^2$ University of Texas at Austin, Austin, Texas 78712 USA}

\begin{abstract}
 A necessary and sufficient condition for linear stability of inviscid
 parallel shear flow is formulated by developing a novel variational
 principle, where the velocity profile is assumed to be monotonic and
 analytic. It is shown that unstable eigenvalues of Rayleigh's equation
 (which is a non-selfadjoint eigenvalue problem) can be associated with positive
 eigenvalues of a certain selfadjoint operator. The stability is therefore determined by maximizing a
quadratic form, which is theoretically and numerically more tractable than directly solving Rayleigh's equation. This variational
 stability criterion is based on the  understandings of Kre\u{i}n 
 signature for continuous spectrum and is  applicable to
 other stability problems of infinite-dimensional Hamiltonian systems.
\end{abstract}

\maketitle

 \section{Introduction}
\label{introduction}

In this paper ideas from the theory of Hamiltonian systems are used to obtain both necessary and sufficient stability conditions by a variational procedure.  The proposed procedure is of general utility, but the treatment here will  be confined  to plane parallel inviscid shear flow (see, e.g., \cite{Howard}).  In this introduction we give a sketch the underlying basis for the procedure in terms of a finite-dimensional Hamiltonian framework, and then place the present contribution in the context of the many previous results for shear flow.

For some Hamiltonian systems, the sign of the curvature of the potential
 energy function provides a necessary and sufficient condition for
 stability.  This is referred to as Lagrange's theorem, which is  the
 crux of many fluid and plasma stability results including  the `$\delta
 W$' energy principle of ideal magnetohydrodynamics (MHD)~\cite{Bernstein}.  For a more general class of Hamiltonian
 systems, definiteness of the Hamiltonian Hessian matrix evaluated at the equilibrium point of interest provides only a sufficient condition for stability.  This is sometimes referred to as Dirichlet's theorem, which is the crux of  many sufficient conditions for stability in the fluid and plasma literature.   The essential reason that Dirichlet's theorem does not provide  a necessary condition for stability is the possible existence of negative energy modes.  Negative energy modes are modes of undamped oscillation,  for which the Hamiltonian is negative when the mode is excited.  For stable (nondegenerate) Hamiltonian systems   of $n$ degrees-of-freedom  the  linear dynamics can be brought by a canonical transformation into the following normal form: 
\begin{equation}
H= \sum_{\alpha=1}^n\sigma_{\alpha}\omega_{\alpha}(p_{\alpha}^2 + q_{\alpha}^2)/2\,,
\label{normal}
\end{equation}
 where $(q_1,q_2, ... q_n;p_1,p_2...,p_n)$ are the  canonically conjugate coordinates,  
$\omega_{\alpha}$  are positive real numbers representing the mode frequencies, and $\sigma_{\alpha}\in\{\pm 1\}$ are the signatures of the mode,  often  called  Kre\u{i}n signature, with $+1$ and $-1$ corresponding to positive and negative energy modes, respectively.   Evidently, systems with both positive and negative energy modes  are  linearly stable, but do  not have a definite Hessian matrix.  (See, e.g.,  \cite{Morrison} for review.)

An advantage afforded by Lagrange's criterion over Dirichlet's is the powerful Rayleigh-Ritz variational method~\cite{Rayleigh0}, which  underlies the MHD and other energy principles.  With this method one need only produce a trial function that makes the Rayleigh quotient negative in order to show instability and, also in this way,  threshold parameter values for the transition to instability can be determined.  When this method is applicable,  it is of great utility because linear  stability conditions for interesting fluid and plasma systems are generally difficult to derive theoretically.  However, it only  applies to a restricted class of Hamiltonian systems for which the eigenvalue problem is self-adjoint, i.e,  systems with steady state bifurcations  to instability through zero frequency that have  pure exponential growth upon transition, and it is known that systems with shear flow are not self-adjoint and have instead  Kre\u{i}n bifurcations \cite{Krein} to overstability, i.e.,   unstable eigenvalues with both nonzero real and imaginary parts.   Such bifurcations are often called Hamiltonian-Hopf 
bifurcations\cite{Krein,Moser,MacKay,Morrison5} and can be viewed as  a resonance between positive and negative energy modes leading to  instability.

Thus we are led to re-examine Dirichlet's principle and seek an alteration that affords the utility of the Rayleigh-Ritz variational method for investigation of  Kre\u{i}n bifurcations.   For finite-dimensional Hamiltonian system written in the normal form of \eqref{normal}, it is evident that $I_{\alpha}=p_{\alpha}^2 + q_{\alpha}^2$ is a constant of motion for each $\alpha$.  From these constants of motion one can construct a  constant of motion with definite Hessian simply by flipping the signs of the signature in the normal form Hamiltonian.  Unfortunately, a priori knowledge of the existence of such a  definite constant of motion is generally not at hand, and one must actually solve the eigenvalue problem in order to be informed of its existence.  This is the essential reason Dirichlet's theorem does not give both necessary and sufficient conditions for stability and a Rayleigh-like variational method is not at hand. 

However, there are two related discoveries that we can exploit to circumvent this deficiency for problems with continuous spectra,  such as  those related to the Vlasov equation, MHD,  and the plane shear flow problem considered here.    The first is the infinite sequence of constants of motion discovered in Ref.~\cite{Oberman}, which were elaborated on in \cite{Case3} and used in the present context in  \cite{Barston}.   The second is the discovery  of a Kre\u{i}n-like signature for the continuous spectra of Vlasov-Poisson equilibria in \cite{Morrison2,Morrison3}, which was applied to plane shear flow in \cite{Balmforth2} and extended to a large class of systems in  \cite{Morrison6}  and  \cite{Hirota,Hirota3}.  These constants of motion in conjunction with the definition of signature allow the construction of a quadratic form, which we will call $Q$, the variation of which can be used in a manner akin Rayleigh's principle for ascertaining stability.  A version of the quadratic form $Q$ was previously given in \cite{Barston}, but it was not used  to obtain sufficient conditions for instability.  We note in passing that the discovery of signature for  the continuous spectrum has also led to rigorous  Kre\u{i}n-like theorems  \cite{Hagstrom1,Hagstrom2,Hagstrom3}, where discrete eigenvalues emerge from  the continuous spectrum, termed Continuum Hamiltonian Hopf bifurcations. (See also \cite{Grillakis,Kapitula} on infinite-dimensional Hamiltonian systems and many other contributions in the recent book \cite{oleg}.)

There have been many significant contributions to the classical plane parallel shear flow problem; thus,  it is important to put our contribution into perspective, which we do here.   The most famous condition is Rayleigh's criterion~\cite{Rayleigh}  that stipulates the existence of an inflection point in the velocity profile is  necessary for  instability, a criterion that was improved by Fj\o rtoft~\cite{Fjortoft}.  These criteria were obtained by direct manipulation of Rayleigh's equation (see  equation \eqref{eq_phi} below), which  governs linear disturbance about a base  shear flow.  The first allusion to Hamiltonian structure for this system appears in the works of Arnold~\cite{Arnold1,Arnold2,Arnold3} who obtained more general sufficient conditions for stability by making use of additional invariants.  This idea was anticipated in the plasma physics literature~\cite{Kruskal}, where the additional invariants were referred to as generalized entropies; today the  generalized entropies are referred to as Casimir invariants  and the general procedure is called the energy-Casimir method (see, e.g., \cite{Hazeltine,Holm,Morrison7,Morrison}).  In \cite{Balmforth2} it was shown explicitly  that all of the above criteria amount to a version of Dirichlet's theorem for this infinite-dimensional Hamiltonian system, where  it was also shown how to explicitly map the system into the infinite-dimensional  version of the normal form of \eqref{normal}.  In this way a signature for the continuous spectrum was first defined for this system, by paralleling the analogous procedure for the Vlasov-Poisson system~\cite{Morrison2,Morrison3}.  Arnold also introduced an important kind of constrained variation  he termed isovortical perturbations (see, e.g., \cite{Arnold}), which are a special case of the dynamically accessible variations of \cite{MP1,MP2} that are generated by Poisson brackets \cite{Morrison8,Morrison}.  Isovortical perturbations together with  a more general Dirichlet-like sufficient condition  for shear flow due to Barston \cite{Barston} play important roles  for obtaining  the results of the present  paper.   

We also note that prior to our variational approach, necessary and sufficient stability conditions were  obtained for certain classes of shear flows using  two other non-Hamiltonian approaches.   One is the perturbation expansion around a  neutrally stable eigenmode, which was pioneered by Tollmien~\cite{Tollmien} and developed by many authors~\cite{Lin,Lin2,ZLin,ZLin2}, and the other is  analysis  based on the Nyquist method~\cite{Rosenbluth,Balmforth} (applied  to magnetohydrodynamics in \cite{Chen}).  These two approaches, however, require detailed probing of Rayleigh's equation to obtain information under specific conditions.   Our variational approach is not only consistent with these earlier results,  but also advantageous  in that we do not have to solve Rayleigh's equation in a rigorous manner.  Namely, we can prove the instability by simply finding a test function (in the appropriate function space) that makes our quadratic form $Q$ positive. This is useful because explicit solutions for Rayleigh's equation are generally not available for a given velocity profile.  Moreover, in numerical calculation, one can obtain  stability boundaries more efficiently from this variational problem (i.e., the maximization of $Q$),  compared to  solving Rayleigh's equation.   We emphasize that our  approach is not limited to  shear flow with rather simple velocity profiles, but the same idea is applicable to various fluid and plasma  stability problems to which it can be of practical use.

Our  paper is organized as follows. In Sec.~\ref{sec:Rayleigh}, Rayleigh's equation is first introduced, and in  Sec.~\ref{sec:theorem} the notion of isovortical variation is reviewed.  Here we describe the quadratic form $Q$,  which  provides   the necessary and sufficient conditions  if the  velocity profile satisfies the assumptions of analyticity and monotonicity.
In particular,  we present the main theorem of this work (Theorem~\ref{th:multiple}), in which  the quadratic form $Q$ is given explicitly.  Then, in Sec.~\ref{sec:proof}, the proof of the main theorem is given.  Here,  we first focus on restricting perturbations to the appropriate function space, and then perform the spectral decomposition in a rigorous manner, which
largely reproduces the well-known spectral properties of Rayleigh's equation (e.g., \cite{satinger}).  
Next, we calculate the signature of $Q$ by applying techniques~\cite{Balmforth2, Hirota,Hirota3} for the action-angle representation of continuous spectrum, where the positive signature of $Q$ indeed predicts the existence of unstable eigenvalues.   Finally,  in Sec.~\ref{sec:proof}, we  show that the function space (the search space on which $Q$
should be maximized) can be extended to a larger one, which is  actually beneficial for  solving the variational problem more effectively.   Section \ref{sec:comparison} contains a demonstration that our variational criterion reproduces  the earlier results of the Nyquist  method~\cite{Rosenbluth,Balmforth} and the perturbation analysis of the  neutral modes~\cite{Tollmien,Lin,Lin2,ZLin,ZLin2}, while Sec.~\ref{sec:numerical}, contains several numerical examples that  
demonstrate  of our theorem.  We summarize in Sec.~\ref{summary}.

\section{Rayleigh equation}
\label{sec:Rayleigh}

We consider the linear stability of inviscid parallel shear flow
$\vec{U}=(0,U(x))$ on a domain $(x,y)\in[-L,L]\times[-\infty,\infty]$
 bounded by two walls at $x=\pm L$, where the flow is assumed to be 
incompressible and two-dimensional.
By introducing the $z$-component of the vorticity disturbance as
 $w(x,t)e^{iky}+{\rm c.c.}$ for a wavenumber $k>0$,
the
linearized vorticity equation is written as
\begin{align}
 i\partial_tw=&kUw+kU''\Ope{G}w\nonumber\\
 =:&k\Ope{L}w,\label{eq_w}
\end{align}
where
the prime ($'$) indicates the $x$ derivative, and
the convolution operator
$\Ope{G}$ is defined by
\begin{align}
 (\Ope{G}w)(x,t)=\int_{-L}^{L}g(x,s)w(s,t)ds,
\end{align}
with a kernel,
\begin{align}
 g(x,s)=\begin{cases}
	 -\frac{\sinh k(s-L)\sinh k(x+L)}{k\sinh 2kL} & x<s\\
	 -\frac{\sinh k(s+L)\sinh k(x-L)}{k\sinh 2kL} & s<x\\
	\end{cases}.
\end{align}
The stream function $\phi(x,t)$ is therefore given by $\phi=\Ope{G}w$ or
$w=\Ope{G}^{-1}\phi=-\phi''+k^2\phi$. 
By assuming an exponential behavior $\phi(x,t)=\hat{\phi}(x)e^{-i\omega
 t}$ with a complex frequency $\omega\in\Bod{C}$,
 the
eigenvalue problem for \eqref{eq_w}
is known as Rayleigh's equation~\cite{Rayleigh};
\begin{align}
 (c-U)(\hat{\phi}''-k^2\hat{\phi})+U''\hat{\phi}=0,
 \label{eq_phi}
\\
 \hat{\phi}(-L)=\hat{\phi}(L)=0,
  \label{eq_bc}
\end{align}
where $c=\omega/k$ is a complex phase speed. If this equation has a nontrivial solution
$\hat{\phi}$ for $c$ with a positive imaginary part, ${\rm Im}\,c>0$, the linearized system
\eqref{eq_w} is
spectrally unstable due to an exponentially growing eigenmode.

\section{Variational stability criterion}\label{sec:theorem}

Hamiltonian structure of the linearized vorticity
equation~\eqref{eq_w} is highly related to its adjoint equation for $\xi(x,t)$~\cite{Balmforth2,Arnold,Hirota};
\begin{align}
 i\partial_t\xi=&kU\xi+k\Ope{G}(U''\xi)\nonumber\\
 =:&k\Ope{L}^*\xi,\label{eq_xi}
\end{align}
where $\Ope{L}^*$ is the adjoint operator of $\Ope{L}$
with respect to the inner product, 
\begin{align}
 \langle
 \xi,\eta\rangle=\int_{-L}^{L}\xi(x)\eta(x)dx\quad\mbox{for}\ \forall\xi,\eta.
\end{align}
Since the relation $U''\Ope{L}^*=\Ope{L}U''$ holds,
$w=-U''\xi$ is found to be a
solution of \eqref{eq_w} if $\xi$ is a solution of \eqref{eq_xi}.
This class of perturbations belonging to the range of $U''$ is said to be isovortical because the vorticity disturbance ($w$) is induced by a displacement ($\xi$) of the fluid while preserving  the conservation law of circulation~\cite{Arnold}. In this manner,   Arnold~\cite{Arnold} derived a constant of motion,
\begin{align}
 \delta^2H=\int_{-L}^{L}\overline{\xi}U''\left[U\xi+\Ope{G}(U''\xi)\right]dx,
\end{align}
where $\overline{\xi}$ denotes the complex conjugate of $\xi$.  Arnold showed that this is  the second variation of the energy  with respect to the isovortical variation, while in \cite{Balmforth2} it was shown that this quantity is in fact the Hamiltonian for the linear Hamiltonian dynamics and there the diagonalizing transformation to action-angle variables was first obtained.    When $U(x)$ has only one inflection point $x_I$ [i.e., $U''(x_I)=0$], then in a frame moving at the velocity
$U_I=U(x_I)$ Arnold replaced  $U$ by  $U-U_I$ in $\delta^2H$ to obtain 
\begin{align}
 \delta^2H_I=&\int_{-L}^{L}\overline{\xi}U''\left[(U-U_I)\xi+\Ope{G}(U''\xi)\right]dx\nonumber\\
 =&\int_{-L}^{L}\overline{w}\left(\frac{U-U_I}{U''}+\Ope{G}\right)w dx.
\end{align}
In \cite{Balmforth2} it was shown explicitly that $\delta^2H_I$ is in fact  the second variation of the full Hamiltonian in the inertial frame boosted by velocity $U_I=U(x_I)$ by adding the appropriate total momentum.

Thus,  we have a version of Dirichlet's theorem, where the shear flow $U(x)$ is stable in the sense of Lyapunov,   if the quadratic form $\delta^2H_I$ is either positive or negative definite, i.e., $\exists \epsilon>0$  such that $\delta^2H_I\ge
\epsilon\langle\overline{w},w\rangle$ or $-\delta^2H_I\ge \epsilon\langle\overline{w},w\rangle$.    If $(U-U_I)/U''>0$ for all $x\in[-L,L]$, $\delta^2H_I$ is clearly
positive definite and hence proves stability, which implies that  this variational criterion of Arnold ~\cite{Arnold1,Arnold2,Arnold3} applies to a larger class of flow profiles than Rayleigh-Fj\o rtoft's stability theorem~\cite{Rayleigh,Fjortoft}.  However, all these criteria, including a generalization by Barston \cite{Barston},  are still sufficient conditions  for stability and, hence, are  indeterminate when  $\delta^2H_I$ is indefinite, since as discussed in Sec.~\ref{introduction} there could be negative energy modes. 

In this work, we  obtain an improved variational criterion, but this requires introducing the following assumptions on $U(x)$.
\begin{flushleft}
 {\bf Assumption} 

(A1) {\it 
$U(x)$ is an analytic and bounded function on $[-L,L]$.
}

 (A2) {\it 
$U(x)$ is strictly monotonic [i.e., $U'(x)\ne0$ for all
$x$] and, if $U''(x_I)=0$ at $x=x_I$, then $U'''(x_I)\ne0$.
}
\end{flushleft}
The last statement implies that the inflection point $x_I$  must
be a simple zero of $U''(x)$ and the sign of $U''(x)$ must change at $x=x_I$.
We expect that it is not difficult to relax 
 these restrictions on $U(x)$ {\it except for} the monotonicity.
To simplify our mathematical arguments we will not pursue generalization in the present work,  but we do  remark upon this point  in our summary of Sec.~\ref{summary}. 

Our main result is that a necessary and sufficient condition for spectral
stability is attained by the following variational criterion.

\begin{theorem}\label{th:multiple}
Let $U(x)$ satisfy (A1) and (A2). Denote the inflection
points of $U$ by $x_{In}$, $n=1,2,\dots,N$, and define
a quadratic form $Q=\langle\xi,\Ope{H}\xi\rangle$ as 
\begin{align}
 Q=\nu\int_{-L}^{L}\xi\prod_{n=1}^N\left[U-U_{In}+U''\Ope{G}\right]U''\xi dx,
\end{align}
where $U_{In}=U(x_{In})$ and  either $\nu=1$
or $\nu=-1$ is chosen such that
\begin{align}
 \nu U''\prod_{n=1}^N(U-U_{In})\le0\label{nu}
\end{align}
holds for all $x\in[-L,L]$.
The equation \eqref{eq_w} is spectrally stable if and
only if
\begin{align}
 Q=\langle\xi,\Ope{H}\xi\rangle\le0
\quad\mbox{for all}\quad\xi\in\Spa{L}^2.\label{condition}
\end{align}
\end{theorem}

Here, $\Spa{L}^2$ denotes the real Hilbert space
on $[-L,L]$ defined by the norm
$\|\xi\|_{\Spa{L}^2}^2=\langle\xi,\xi\rangle$. 
Note that the function space for disturbances $w$ and $\xi$ is originally
the complex Hilbert space $\Spa{L}^2+i\Spa{L}^2$ defined by
the norm $\|\xi\|_{\Spa{L}^2+i\Spa{L}^2}^2=\langle\overline{\xi},\xi\rangle$.
Since $\Ope{H}$ is a real selfadjoint operator,
the condition \eqref{condition} can be naturally replaced by $Q=\langle\overline{\xi},\Ope{H}\xi\rangle\le0$ for all
$\xi\in\Spa{L}^2+i\Spa{L}^2$.

This $Q=\langle\overline{\xi},\Ope{H}\xi\rangle$ is
equivalent to the constant of motion derived by Barston [11] (except for
the coefficient $\nu$), and $Q=-\nu\delta^2H_I$ for the case of single inflection point. Hence,
 our theorem claims that Arnold-Barston's stability criteria (Dirichlet sufficient stability conditions) are in fact necessary and sufficient when $U(x)$ satisfies (A1) and (A2).

We remark that
$Q=\langle\overline{\xi},\Ope{H}\xi\rangle$ no longer represents the
second variation of the energy for the case of multiple infection points.
Actually, it belongs to the  class of infinite number of
constants of motion introduced in \cite{Oberman,Case3,Barston}.

\begin{proposition}
Let $f(c)$ be any real polynomial of $c\in\Bod{R}$. Then,
\begin{align}
 Q_f=\int_{-L}^{L}\overline{\xi}U''f(\Ope{L}^*)\xi dx
=\int_{-L}^{L}\overline{\xi}f(\Ope{L})(U''\xi)dx
\in\Bod{R},
\end{align} 
is a constant of motion for the equation \eqref{eq_xi}.
\end{proposition}
\begin{proof}
Using $U''\Ope{L}^*=\Ope{L}U''$ and $\Ope{L}^*f(\Ope{L}^*)
=f(\Ope{L}^*)\Ope{L}^*$, we can directly show that $Q_f$ is real and $dQ_f/dt=0$.
\end{proof}
Therefore, we have specifically chosen $f(c)=\nu\prod_{n=1}^N(c-U_{In})$
to generate $Q$ of Theorem~\ref{th:multiple}.

The proof of Theorem~\ref{th:multiple} is given in the next section and our strategy is as
 follows. First, we reduce the function space $\Spa{L}^2+i\Spa{L}^2$
 to a smaller one $X+iX$, to which unstable eigenfunctions must belong.
Then, we decompose the spectrum
$\sigma\subset\Bod{C}$ of the operator $k\Ope{L}^*$ into the
neutrally stable part $\sigma_c\subset\Bod{R}$ (which is mostly a
continuous spectrum) and the remaining part $\sigma\backslash\sigma_c\subset\Bod{C}\backslash\Bod{R}$
(which is a set of pairs of growing and damping eigenvalues).
By proving that $Q\le 0$ for all the neutrally stable disturbance $\xi$ belonging to
$\sigma_c$, we will claim that $Q>0$ for some $\xi\in X$ indicates the
existence of at least one unstable eigenvalue,
$\omega\in\sigma\backslash\sigma_c$ that has a growth rate ${\rm Im}\,\omega>0$.

\section{Proof of Theorem~\ref{th:multiple}}\label{sec:proof}

\subsection{Reduction to isovortical disturbance}\label{sec:isovortical}

For
the purpose of seeking unstable eigenmodes, we restrict the function
space of disturbance to $X+iX$, where
\begin{align}
 X=\Spa{H}_0^1\cap\Spa{H}^2.
\end{align}
As usual, we denote by $\Spa{H}^n$ the real Sobolev space on $[-L,L]$;
\begin{align}
 \Spa{H}^n=\left\{\xi\in\Spa{L}^2\ \bigg|\ \sum_{j\le n}\|\partial_x^j\xi\|_{\Spa{L}^2}<\infty\right\}
\end{align}
and $\Spa{H}_0^1$ denotes the subspace of $\Spa{H}^1$ in which
 the
 boundary conditions, $\xi(-L)=\xi(L)=0$, are imposed on $\xi\in\Spa{H}_0^1$.
The restriction from $\Spa{L}^2+i\Spa{L}^2$ to $X+iX$ is feasible when $U(x)$ is a sufficiently smooth function. In this work, we simply assume (A1) is sufficient for the following:

\begin{proposition}\label{th:isovortical}
Let $U(x)$ satisfy (A1).
 Given the initial condition $\xi(x,0)=\xi_0(x)\in X+iX$,
        the solution $\xi(x,t)$ of \eqref{eq_xi} belongs to
       $X+iX$ for all
       $t$. Moreover, $w=-U''\xi\in X+iX$ is a solution of \eqref{eq_w}.
\end{proposition}
\begin{proof}
 By noting that $\Ope{G}:\Spa{L}^2\rightarrow X$
       is one to one and onto, we find that $\Ope{L}^*$ is a
 bounded operator on
 $X+iX$ and, hence, the solution $\xi=e^{-ik\Ope{L}^*t}\xi_0$ belongs to $X+iX$ for all $t$.
Using the property $U''\Ope{L}^*=\Ope{L}U''$, it is obvious that $w=-U''\xi$ is
 automatically a solution of \eqref{eq_w} and also belongs to $X+iX$.
\end{proof}

When $U''$ vanishes somewhere on $[-L,L]$, the function space of
$w=-U''\xi$ is further restricted to the range of $U''$ (i.e., the isovortical
disturbance). We can find that all unstable eigenfunctions
       must belong to this space as follows.

\begin{proposition}
Let $U(x)$ satisfy (A1).
The equation \eqref{eq_w} is spectrally stable if and only if the
 adjoint equation \eqref{eq_xi} for $\xi\in X+iX$ is spectrally stable. 
\end{proposition} 
\begin{proof}
If $c\in\Bod{C}$ and $\hat{w}=-\hat{\phi}''+k^2\hat{\phi}\in\Spa{L}^2+i\Spa{L}^2$ satisfy Rayleigh's equation
with a growth rate ${\rm Im}\, c>0$, then $(c-U)\ne0$ holds everywhere and 
\begin{align}
\hat{\xi}=-\frac{1}{c-U}\Ope{G}\hat{w} \ \in X+iX
\end{align}
is found to be an eigenfunction of the adjoint equation \eqref{eq_xi}
with the same eigenvalue $c$. Hence, the adjoint equation on $X+iX$ is also spectrally unstable.

Conversely, if $c$ and $\hat{\xi}\in X+iX$ satisfy the adjoint
eigenvalue problem with ${\rm Im}\,c>0$, then $U''\hat{\xi}$ is not
identically zero and $\hat{w}=-U''\hat{\xi}$ satisfies the Rayleigh equation with the same $c$.  
\end{proof}

\subsection{Spectral decomposition}\label{sec:spectral}

\begin{figure}
\centering\includegraphics[width=6cm]{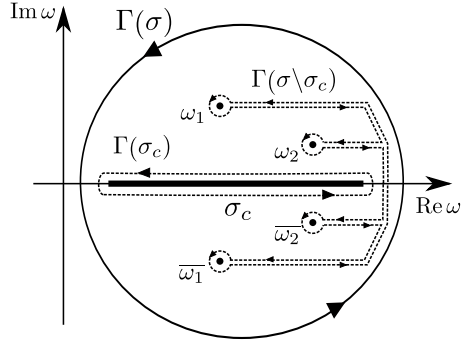}
 \caption{Schematic view of contour integral
 $\Gamma(\sigma)=\Gamma(\sigma_c)\cup\Gamma(\sigma\backslash\sigma_c)$
 in the case of $\sigma\backslash\sigma_c=\{\omega_1,\omega_2,\overline{\omega_1},\overline{\omega_2}\}$.
}
\label{fig:contours}
\end{figure}

Next, we investigate the spectrum $\sigma\subset\Bod{C}$ of the operator $k\Ope{L}^*$.
For a given initial condition $\xi(x,0)=\xi_0(x)\in X+iX$, let
$\Xi(x,\Omega)\in X+iX$ be the solution of 
\begin{align}
 (\Omega-k\Ope{L}^*)\Xi(x,\Omega)=&\xi_0(x),\label{eq_Xi}
\end{align}
for $\Omega\in\Bod{C}\backslash\sigma$.
Then, the solution of \eqref{eq_xi} is formally represented by the
Dunford integral (or the inverse Laplace transform),
\begin{align}
 \xi(x,t)=\frac{1}{2\pi i}\oint_{\Gamma(\sigma)}\Xi(x,\Omega)e^{-i\Omega t}d\Omega,\label{Dunford}
\end{align}
where $\Gamma(\sigma)$ is a path of integration that encloses
all the spectrum $\sigma\subset\Bod{C}$ of $k\Ope{L}^*$ counterclockwise
(as shown in FIG.~\ref{fig:contours}). 
In terms of $\Phi(x,\Omega)=-(\Omega/k-U)\Xi(x,\Omega)$, the equation
\eqref{eq_Xi} is transformed into
\begin{align}
\Ope{E}(\Omega)\Phi(x,\Omega)=\frac{1}{k}(\xi_0''-k^2\xi_0),
\label{eq_Phi}\\
\Phi(-L,\Omega)=\Phi(L,\Omega)=0,
\label{eq_BC}
\end{align}
where 
\begin{align}
 \Ope{E}(\Omega)=-\frac{\partial^2}{\partial x^2}+k^2-\frac{kU''}{\Omega-kU}.
\end{align}
 Suppose that we have solved
\begin{align}
 \Ope{E}(\Omega)\Phi_<(x,\Omega)=0,\quad \Phi_<(-L,\Omega)=0,\quad
 \Phi_<'(-L,\Omega)=1,\label{Phi_<}
\end{align}
and
\begin{align}
 \Ope{E}(\Omega)\Phi_>(x,\Omega)=0,\quad \Phi_>(L,\Omega)=0,\quad
\Phi_>'(L,\Omega)=-1,
\end{align}
to obtain two linearly independent solutions $\Phi_<(x,\Omega)$ and
$\Phi_>(x,\Omega)$.
Then, by
using the method of Green's function, the solution of \eqref{eq_Phi} and  \eqref{eq_BC} can be expressed as
\begin{align}
 \Phi(x,\Omega)=\frac{1}{W(\Omega)}\int_{-L}^{L}\Phi_G(x,s,\Omega)\frac{1}{k}[\xi_0''(s)-k^2\xi_0(s)]ds,
\end{align}
where
\begin{align}
 \Phi_G(x,s,\Omega)=\Phi_>(s,\Omega)\Phi_<(x,\Omega)Y(s-x)
+\Phi_<(s,\Omega)\Phi_>(x,\Omega)Y(x-s),
\end{align}
with $Y(x)$ being the Heaviside function, and
\begin{align}
 W(\Omega)=&-\Phi_<(x,\Omega)\Phi_>'(x,\Omega)+\Phi_<'(x,\Omega)\Phi_>(x,\Omega)\nonumber\\
=&\Phi_>(-L,\Omega)\nonumber\\
=&\Phi_<(L,\Omega)
\end{align}
is the Wronskian.  When  $\Omega$ avoids the range of $kU$, denoted by
\begin{align}
 \sigma_c=\{kU(x)\in\Bod{R}\ |\ x\in[-L,L]\},
\end{align}
the operator $\Ope{E}(\Omega)$ is
 non-singular and both $\Phi_<$ and $\Phi_>$ are regular functions of
 $\Omega$. Therefore, the spectrum $\sigma$ of the operator $k\Ope{L}^*$ is composed of a continuous spectrum $\sigma_c$ and some eigenvalues
 $\omega_j\in\Bod{C}, j=1,2,\dots$ (i.e.,
 point  spectra) that  satisfy $W(\omega_j)=0$. 
Due to the property $W(\Omega)=\overline{W(\overline{\Omega})}$, the
 eigenvalues always exist as pairs of growing ($\omega_j$) and damping
 ($\overline{\omega_j}$) ones when ${\rm Im}\,\omega_j>0$,  a spectral property of Hamiltonian systems (cf.\  \cite{Hagstrom1,Hagstrom2,Hagstrom3}). 

For the purpose of showing the existence or nonexistence of such
eigenvalues, we will frequently use  the following Lemma.

 \begin{lemma}\label{lem:Sturm}
Let $U(x)$ satisfy (A1) and (A2). For any $\omega\in\Bod{R}$,
the general solution $\Phi(x,\omega)$ of
 $\Ope{E}(\omega)\Phi(x,\omega)=0$ has at most one zero on
each of the following intervals: 

(i) $[-L,L]$ if there is no critical layer, i.e., $\nexists\, \omega\in \Bod{R}$  that satisfies $\omega=kU(x)\,  \forall x\in [-L,L]$, 

(ii) $[-L,x_c]$ and $[x_c,L]$ if there is a critical
 layer $x_c\in[-L,L]$ that satisfies $\omega=kU(x_c)$.

\end{lemma}
\begin{proof}

(i)  A consequence of  $\Ope{E}(\omega)\Phi(x,\omega)=0$ is the following identity (see the Appendix of \cite{Balmforth}):
\begin{align}
\left[\Phi\left(\Phi' -\frac{U'\Phi}{U-c}\right)\right]^{x_2}_{x_1} =\int_{x_1}^{x_2}\left[
\left(\Phi' -\frac{U'\Phi}{U-c}\right)^2 + k^2 \Phi^2
\right]\, dx,
\label{identity}
\end{align}
which is valid for any solution $\Phi$ and subinterval
 $[x_1,x_2]\subseteq[-L,L]$.  This identity follows directly from
 Rayleigh's equation by multiplying by $\Phi$, manipulating, and
 integrating. 
If there are two zeros, say $x_{z1}$ and $x_{z2}$, that satisfy $\Phi(x_{z1},\omega) = \Phi(x_{z2},\omega)=0$,
choosing
 $x_1=x_{z1}$ and $x_2=x_{z2}$ implies
\begin{align}
\int_{x_{z1}}^{x_{z2}}\left[
\left(\Phi' -\frac{U'\Phi}{U-c}\right)^2 + k^2 \Phi^2
\right]\, dx
=0,
\label{identity2}
\end{align}
which requires $\Phi$ to be  the trivial solution $\Phi\equiv 0$.
Thus, any (nontrivial) solution has at most one zero on $[-L,L]$.

(ii) Consider the interval $[x_c,L]$ (the same argument goes for
 $[-L,x_c]$). The identity \eqref{identity} again holds for
 $x_c<x_1<x_2\le L$ and, hence, there is at most one zero on
 $(x_c,L]$. In the neighborhood of $x_c$, the solution $\Phi$ is
 expressed by a linear combination of the Frobenius series solutions, in
 which 
$\Phi(x_c,\omega)$ is always bounded [see \eqref{eq_bnded} below]. If $\Phi(x_c,\omega)\ne0$, the Lemma is automatically
 proven [although both sides of \eqref{identity} go to infinity as
 $x_1\rightarrow x_c$]. If $\Phi(x_c,\omega)=0$, the identity
 \eqref{identity} is still valid for $x_1=x_c$ and we can prove that
 there is no other  zero ($x_{z2}$) except for $x_{z1}=x_c$ on $[x_c,L]$, in the same
 manner as (i).

\end{proof}

It is well-known from Tollmien's argument~\cite{Tollmien,Lin,Lin2} that,
as $k$ varies, an unstable eigenvalue ${\rm Im}\,c>0$ of Rayleigh's equation emerges through a
neutrally stable eigenvalue $c=U(x_I)\in\Bod{R}$ where $U''(x_I)=0$. In
other words, 
the neutrally stable eigenvalue $c\in\Bod{R}$ can exist only if $U''(x_c)=0$ at the corresponding critical layer $x_c=U^{-1}(c)$.
Unfortunately, this argument is not always true especially for nonmonotonic profiles
of $U(x)$ (see Ref.~\cite{ZLin,ZLin2,Balmforth} for mathematical justification in the shear flow context and \cite{Hagstrom1,Hagstrom2} for a discussion of $k\neq 0$ bifurcations in the Vlasov context).
In this work, we simply assume the monotonicity (A2) and
verify Tollmien's argument as follows.

\begin{proposition}\label{th:discrete}
 Let $U(x)$ satisfy (A1) and (A2). Denote the inflection points of $U$
 by $x_{In}$, $n=1,2,\dots,N$ and define $U_{In}:=U(x_{In})$. Then,
the function $W(\omega\pm i0)$ of $\omega\in\Bod{R}$ can vanish only at
 $\omega=kU_{In}$, $n=1,2,\dots,N$, and moreover
\begin{align}
 \lim_{\Omega\rightarrow\omega\pm i0}\frac{\prod_{n=1}^N(\Omega-kU_{In})}{W(\Omega)}<\infty.\label{simple}
\end{align}
\end{proposition}
 
\begin{proof}
For $\omega\in\Bod{R}\backslash\sigma_c$, there is no critical layer
and, from Lemma~\ref{lem:Sturm}, 
the solution $\Phi_<(x,\omega)$ does not have zero on $[-L,L]$ except for $x=-L$.
 Hence, 
$W(\omega)=\Phi_<(L,\omega)\ne0$.

For $\omega\in\sigma_c$, there is only one critical layer
 $x_c=U^{-1}(\omega/k)$.
Since $U(x)$ is analytic, $\Phi_<(x,\Omega)$ can be expressed by
 a linear combination of the Frobenius series solutions (so-called Tollmien's
 inviscid solutions) around $x_c$.
By taking account of the branch cut of the logarithmic function, it is written  in the
 limit $\Omega\rightarrow\omega\pm i0$ as
\begin{align}
 \Phi_<(x,\omega\pm i0)=&C_r(\omega)\Phi_1(x,\omega)\nonumber\\
& +C_s(\omega)\left\{\Phi_2(x,\omega)
+\frac{U''(x_c)}{U'(x_c)}\Phi_1(x,\omega)
\left[\log|x-x_c|\mp\pi iY(x-x_c)\right]
\right\},
\end{align}
where $\Phi_1(x,\omega)$ and $\Phi_2(x,\omega)$ are real and regular
 functions with $\Phi_1(x_c,\omega)=\Phi_2'(x_c,\omega)=0$ and
 $\Phi_1'(x_c,\omega)=\Phi_2(x_c,\omega)=1$ (see Ref.~\cite{Drazin}).
From the definition \eqref{Phi_<}, the coefficients $C_r(\omega)$ and $C_s(\omega)$ are
 found to be real and so is
 $\Phi_<(x,\omega\pm i0)$ on $[-L,x_c]$ since $Y(x-x_c)\equiv0$ for $x<x_c$.
Lemma~\ref{lem:Sturm} shows that $\Phi_<(x,\omega\pm i0)$ has no zero on
 $[-L,x_c]$ other than $x=-L$ and hence $C_s(\omega)>0$ [since
 $\Phi_1(x_c,\omega)=0$].

If $\omega\ne kU_{In}$ (i.e., $x_c\ne x_{In}$), 
 then $U''(x_c)\ne0$ holds and $\Phi_<(x,\omega\pm
 i0)$ possesses the imaginary part on
 $[x_c,L]$. Lemma~\ref{lem:Sturm} again shows that this imaginary part
 has no zero on $[x_c,L]$ other than $x=x_c$.
 Therefore, we
 conclude that
 ${\rm Im}\,W(\omega\pm i0)={\rm Im}\,\Phi_<(L,\omega\pm i0)\ne0$ for
 all $\omega\in\Bod{R}\backslash\{kU_{In}|n=1,2,\dots,N\}$. 

If  $\omega=kU_{In}$, then  $\Phi_<(x,kU_{In})$ is a real and regular
 function on the whole domain $[-L,L]$ due to $U''(x_c)=0$ and has at most one zero on $[x_c,L]$.
 Therefore, $W(kU_{In})=0$ may occur only when
 this zero corresponds to $x=L$, and as such the zero must
 be simple, namely, \eqref{simple} holds.
\end{proof}

Besides the singularity stemming from the zeros of $W(\Omega)$, $\Phi(x,\Omega)$
has also the following essential singularity along the continuous spectrum $\sigma_c$.

\begin{proposition}\label{th:continuous}
 Let $U(x)$ satisfy (A1) and (A2). For all $\xi_0\in X+iX$ and
 $\omega\in\Bod{R}$,
\begin{align}
 \Psi(x,\omega\pm i0)\in\Spa{H}^1+i\Spa{H}^1,
\end{align}
where
\begin{align}
 \Psi(x,\Omega):=\Phi(x,\Omega)\prod_{n=1}^N(\Omega/k-U_{In}).
\end{align}
\end{proposition}
\begin{proof}
For any fixed $s\in[-L,L]$, the function $\partial\Phi_G/\partial
 x(x,s,\omega\pm i0)$ is regular almost everywhere except that it has a logarithmic singularity,
 $\log|x-x_c|$, and discontinuities, $Y(x-x_c)$ and $Y(x-s)$. Hence, 
\begin{align}
 \int_{-L}^{L}\left|\frac{\partial\Phi_G}{\partial x}(x,s,\omega\pm i0)\right|dx<\infty.
 \label{eq_bnded}
\end{align}
Since $\xi_0''-k^2\xi_0\in\Spa{L}^2+i\Spa{L}^2$, 
 the following convolution integral is also square-integrable,
\begin{align}
 \lim_{\Omega=\omega\pm i0}\left[W(\Omega)\Phi'(x,\Omega)\right]=
 \int_{-L}^{L}\frac{\partial\Phi_G}{\partial x}(x,s,\omega\pm i0)\frac{1}{k}[\xi_0''(s)-k^2\xi_0(s)]ds\in\Spa{L}^2+i\Spa{L}^2,
\end{align}
that is,
\begin{align}
 \lim_{\Omega=\omega\pm i0}\left[W(\Omega)\Phi(x,\Omega)\right]\in\Spa{H}^1+i\Spa{H}^1.
\end{align}
In combination with \eqref{simple}, the proposition is proven.
\end{proof}

\subsection{Signature of $Q$}

Now, let us substitute the expression \eqref{Dunford} into the quadratic
form $Q=\langle\overline{\xi},\Ope{H}\xi\rangle$, which is a  constant of motion for the equation \eqref{eq_xi}. 
By using the
property of the resolvent $(\Omega-k\Ope{L}^*)^{-1}$ (see
Ref.~\cite{Hirota,Hirota3}), we obtain
\begin{align}
 Q=&\frac{1}{2\pi i}
\oint_{\Gamma(\sigma)}
h(\Omega)
d\Omega
\end{align}
where $h:\Bod{C}\rightarrow\Bod{C}$ is given by
\begin{align}
 h(\Omega)
=\nu\int_{-L}^{L}\overline{\xi_0}
\left[-\frac{kU''}{\Omega-kU}\Psi(x,\Omega)\right]dx.
\end{align}
Upon decomposing the spectrum $\sigma\subset\Bod{C}$ into
$\sigma_c\subset\Bod{R}$ and others
$\sigma\backslash\sigma_c\subset\Bod{C}\backslash\Bod{R}$ and, 
accordingly,  deforming  the contour $\Gamma(\sigma)$ into
$\Gamma(\sigma_c)$ and $\Gamma(\sigma\backslash\sigma_c)$ (see FIG.~\ref{fig:contours}),  we obtain
$Q=Q|_{\sigma_c}+Q|_{\sigma\backslash\sigma_c}$  with 
\begin{align}
Q|_{\sigma_c}=&\frac{1}{2\pi
 i}\oint_{\Gamma(\sigma_c)}h(\Omega)d\Omega
=\int_{\sigma_c}\hat{h}(\omega)d\omega,
\end{align}
where
\begin{align}
 \hat{h}(\omega)=\frac{1}{2\pi i}[-h(\omega+i0)+h(\omega-i0)].
\end{align}
The existence of this limit for
all $\omega\in\sigma_c$, is guaranteed by
Proposition~\ref{th:continuous} and $\xi_0\in X+iX$.

To observe the signature of the function $\hat{h}(\omega)$ more explicitly,
we rewrite $h(\Omega)$ as
\begin{align}
 h(\Omega)=&\frac{\nu k}{p(\Omega)}\int_{-L}^{L}
\overline{\Psi(x,\overline{\Omega})}\Ope{E}(\Omega)\Psi(x,\Omega)
dx
-\nu \frac{p(\Omega)}{k}\int_{-L}^{L}(|\xi_0'|^2+k^2|\xi_0|^2)
dx,
\end{align}
where we have put $p(\Omega)=\prod_{n=1}^N(\Omega/k-U_{In})$.
Since $p(\Omega)$ is an regular function of $\Omega$,  we may neglect
the second term on the right hand side when calculating
$\hat{h}(\omega)$. As shown in Proposition 10 of Ref.~\cite{Hirota}, the first term is
further transformed into
\begin{align}
k\int_{-L}^{L}
\overline{\Psi(x,\overline{\Omega})}\Ope{E}(\Omega)\Psi(x,\Omega)
dx
=&-k\langle\overline{F(x,\Omega)},\Ope{E}(\Omega)F(x,\Omega)\rangle
-k\langle\overline{G(x,\Omega)},\Ope{E}(\Omega)G(x,\Omega)\rangle\nonumber\\
&+p(\Omega)\langle\overline{G(x,\Omega)},\xi_0''-k^2\xi_0\rangle
+p(\Omega)\langle\overline{\xi_0''-k^2\xi_0},G(x,\Omega)\rangle,
\end{align}
where
\begin{align}
 F(x,\Omega)=\frac{1}{2}[\Psi(x,\Omega)-\Psi(x,\overline{\Omega})],\\
 G(x,\Omega)=\frac{1}{2}[\Psi(x,\Omega)+\Psi(x,\overline{\Omega})].
\end{align}
In the limit of $\Omega\rightarrow\omega\pm i0$, the relations
$F(\omega+i0)=-F(\omega-i0)$ and $G(\omega+i0)=G(\omega-i0)$ hold.
Using the formula,
\begin{align}
 -\Ope{E}(\omega+i0)+\Ope{E}(\omega-i0)=&\frac{kU''(x)}{\omega+i0-kU(x)}-\frac{kU''(x)}{\omega-i0-kU(x)}\nonumber\\
=&-2\pi i\frac{U''(x_c)}{|U'(x_c)|}\delta(x-x_c),
\end{align}
we finally obtain
\begin{align}
 \hat{h}(\omega)=&\frac{\nu kU''(x_c)}{p(\omega)|U'(x_c)|}
\left[|F(x_c,\omega+i0)|^2+|G(x_c,\omega+i0)|^2\right],
\end{align}
where $x_c=U^{-1}(\omega/k)$ should be read as a function of
$\omega$.
According to the definition \eqref{nu} of $\nu$,
this expression indicates that $\hat{h}(\omega)$
is negative for all $\omega\in\sigma_c$ and we conclude that
\begin{align}
 Q|_{\sigma_c}
=\int_{\sigma_c}\hat{h}(\omega)d\omega\le0
\end{align}
for all solutions $\xi\in X+iX$ of \eqref{eq_xi} with initial data
$\xi_0\in X+iX$.

If $Q=\langle\overline{\xi},\Ope{H}\xi\rangle>0$ for some $\xi\in X+iX$, then $\sigma\backslash\sigma_c$ must
not be null and there exists at least one pair of complex
eigenvalues, say $\omega_j$ and $\overline{\omega_j}$ with ${\rm
Im}\,\omega_j>0$, which correspond to growing and damping modes, respectively.
Since $\Ope{H}$ is actually a real self-adjoint operator, the
condition $Q=\langle\xi,\Ope{H}\xi\rangle>0$ for some $\xi\in X$
comes to the same conclusion.

Conversely, if there exist several pairs of complex eigenvalues
$\sigma\backslash\sigma_c=\{\omega_j,\overline{\omega}_j|j=1,2,\dots\}$,
the constant of motion $Q$ must be indefinite in the corresponding
eigenspaces, as shown by \cite{Whittaker,Moser,MacKay}  for a  Hamiltonian function.
Indeed, the solution $\xi$ is subject to the following modal decomposition;
\begin{align}
 \xi=\sum_j(a_j\hat{\xi}_je^{-i\omega_jt}+b_j\overline{\hat{\xi_j}}e^{-i\overline{\omega_j}t})+\dots,
\end{align}
where $\hat{\xi}_j$ is the eigenfunction for $\omega_j$ and
$a_j,b_j\in\Bod{C}$ are the mode amplitudes which depend on $\xi_0$. By
substituting this expression into $Q$
and using the orthogonality relations,
\begin{align}
 \langle\hat{\xi}_l,U''\hat{\xi}_j\rangle=
\langle\overline{\hat{\xi}_l},U''\hat{\xi}_j\rangle=
\langle\overline{\hat{\xi}_j},U''\hat{\xi}_j\rangle=0
\quad(l\ne j),\quad
\langle\hat{\xi}_j,U''\hat{\xi}_j\rangle\ne0,
\end{align}
 we obtain
\begin{align}
 Q|_{\sigma\backslash\sigma_c}=\nu\sum_j\left[a_j\overline{b_j}p(\omega_j)\langle\hat{\xi}_j,U''\hat{\xi}_j\rangle
+{\rm c.c.}\right],
\end{align}
whose sign is clearly indefinite. For example, by setting either
$(a_j,b_j)=(1,1)$ or $(a_j,b_j)=(1,-1)$, we can make
$Q|_{\sigma\backslash\sigma_c}>0$.
Thus, we have proven that the equation \eqref{eq_w} is spectrally stable if and
only if
$Q\le0$
for all $\xi\in X$.

\subsection{Extension of search space}\label{sec:extension}

Our remaining task is to extend the search space from $X$ to $\Spa{L}^2$.
Maximization of $Q=\langle\xi,\Ope{H}\xi\rangle$ on $\Spa{L}^2$ is, in practice, more tractable than
that on $X$, since the variational problem $\lambda_{\rm max}=\max
Q/\|\xi\|^2_{\Spa{L}^2}$ simply searches the maximum eigenvalue
$\lambda_{\rm max}$ of the selfadjoint operator $\Ope{H}$. 
Let us consider the eigenvalue
problem $(\lambda-\Ope{H})\hat{\xi}=0$ and rewrite  $\Ope{H}$ in the form of
\begin{align}
 \Ope{H}=\nu U''\prod_{n=1}^N(U-U_{In})+\Ope{R},
\end{align}
where $\Ope{R}$ represents the sum of all operators that involve at
least one multiplication of $\Ope{G}$ and hence
$\Ope{R}:\Spa{L}^2\rightarrow X$.
It follows from the
condition \eqref{nu} that $\Ope{H}$ has a continuous spectrum for the
negative side, $\min[\nu U''\prod_{n=1}^N(U-U_{In})]
\le\lambda\le0$.
 On the other hand, for $\lambda>0$, the eigenvalue
problem is non-singular and can be rewritten as
\begin{align}
 \hat{\xi}=\frac{1}{\lambda-\nu U''\prod_{n=1}^N(U-U_{In})}\Ope{R}\hat{\xi}.
\end{align}
Since $\Ope{R}\hat{\xi}\in X$, this eigenfunction $\hat{\xi}$ inevitably belongs to $X$. 
If $\Ope{H}$ has such a positive discrete eigenvalue, the
 corresponding eigenfunction $\hat{\xi}\in X$ directly proves the instability $Q=\langle\hat{\xi},\Ope{H}\hat{\xi}\rangle>0$. Conversely, if $Q\le0$ for all $\xi\in\Spa{L}^2$, then obviously
$Q\le0$ for all $\xi\in X\subset\Spa{L}^2$. Therefore,
we may replace the search space $X$
 by $\Spa{L}^2$;  thus, the proof of Theorem~\ref{th:multiple} is completed.

We can further extend this idea as follows: 
\begin{corollary}\label{th:extension}
 The stability condition \eqref{condition} in Theorem~\ref{th:multiple}
can be replaced by
\begin{align}
 Q=\langle w,\Ope{H}_{\rm v}w\rangle\le0
\quad\mbox{for all}\quad w\in\Spa{L}^2,
\end{align}
where $w=-U''\xi$ and, hence,
\begin{align}
 \Ope{H}_{\rm v}=\frac{\nu}{U''}\prod_{n=1}^N\left[U-U_{In}+U''\Ope{G}\right].
\end{align}
\end{corollary}

\begin{proof}
Since $(\nu/U'')\prod_{n=1}^N(U-U_{In})<0$ follows from the assumptions, the operator $\Ope{H}_{\rm v}$ is found to be bounded; $\exists
 C>0$ such that $\langle w,\Ope{H}_{\rm v}w\rangle<C\|w\|_{\Spa{L}^2}^2$ for all
 $w\in\Spa{L}^2$.
Suppose that we find a function $\hat{w}\in\Spa{L}^2$ that 
 makes $Q$ positive;
\begin{align}
0<\frac{\langle\hat{w},\Ope{H}_{\rm v}\hat{w}\rangle}{\|\hat{w}\|_{\Spa{L}^2}^2}<C.
\end{align}
Then, consider a sequence $\xi_m\in\Spa{L}^2$, $m=1,2,\dots,\infty$, that
 satisfies $\|\hat{w}+U''\xi_m\|_{\Spa{L}^2}\rightarrow0$ as
$m\rightarrow\infty$. 
Since $\langle\xi_m,\Ope{H}\xi_m\rangle\rightarrow\langle\hat{w},\Ope{H}_{\rm v}\hat{w}\rangle$ as $m\rightarrow\infty$,
$\langle\xi_m,\Ope{H}\xi_m\rangle$ also becomes positive when $m$ is sufficiently large.

On the
other hand, if $\langle w,\Ope{H}_{\rm v}w\rangle\le0$ for all $w\in\Spa{L}^2$, then obviously
$\langle\xi,\Ope{H}\xi\rangle=\langle U''\xi,\Ope{H}_{\rm v}U''\xi\rangle\le0$ for all $\xi\in\Spa{L}^2$.
\end{proof}

We will actually adopt Corollary~\ref{th:extension} in the subsequent
sections, because this variational problem for $w\in\Spa{L}^2$ is more
beneficial than that for $\xi\in\Spa{L}^2$,  both analytically and
numerically. This fact is evident from the corresponding eigenvalue problem;
\begin{align}
 (\lambda -\Ope{H}_{\rm v})\hat{w}=0.\label{eigen_multiple}
\end{align}
The operator $\Ope{H}_{\rm v}$, which is  again written as
\begin{align}
 \Ope{H}_{\rm v}=&\frac{\nu}{U''}\prod_{n=1}^N(U-U_{In})+\frac{1}{U''}\Ope{R}\frac{1}{U''},
\end{align}
 has a continuous spectrum, but
 it is remarkable that 
the upper edge of this continuous spectrum,
$\lambda_u=\max[(\nu/U'')\Pi_{n=1}^N(U-U_{In})]$,
 is separated from the origin ($\lambda_u<0$).
Owing to this property, the variational problem for $w\in\Spa{L}^2$ is useful for investigating the stability boundary
at $\lambda=0$ without suffering from any singularity.

\section{Comparison with existing results}\label{sec:comparison}

In this section, we  explore several alternative representations of our variational stability criterion
by assuming that we have somehow {\it solved}  Rayleigh's
equation under  specific conditions.  As a consequence  of this exploration,
we reproduce existing stability theorems and 
gain a clear-cut understanding of the onset of instability.

\subsection{Single inflection point}

First consider the case of a single inflection point with the condition
$(U-U_I)/U''<0$ for all $x\in[-L,L]$, since the opposite case $(U-U_I)/U''>0$ is always stable. According to Corollary~\ref{th:extension}, we maximize $Q$ with
respect to
$w\in\Spa{L}^2$, where the corresponding
eigenvalue problem \eqref{eigen_multiple} is simply
\begin{align}
 \lambda \hat{w}=\frac{U-U_I}{U''}\hat{w}+\Ope{G}\hat{w}.
\end{align}
We are interested in whether a positive eigenvalue $\lambda>0$ exists or
not, for its existence is the necessary and
sufficient condition for instability. 
By focusing on $\lambda>\lambda_u$ where $\lambda_u=-\min[(U_I-U)/U'']<0$,
the eigenvalue problem is transformed into 
\begin{align}
 \hat{\phi}''-k^2\hat{\phi}+\frac{1}{\lambda+(U_I-U)/U''}\hat{\phi}=0,\label{eigen}\\
 \hat{\phi}(-L)=\hat{\phi}(L)=0,\label{boundary}
\end{align}
using $\hat{\phi}=\Ope{G}\hat{w}$.
Since this is a Sturm-Liouville problem, the general solution
$\hat{\phi}$ becomes less oscillatory everywhere on $[-L,L]$ as  the 
two parameters $\lambda>\lambda_u$ and $k>0$ increase.
It follows that the eigenvalue $\lambda$ is bounded by
\begin{align}
\lambda< \frac{1}{k^2}+\lambda_u.
\end{align}
If $k^2>-\lambda_u^{-1}$, no positive eigenvalue $\lambda>0$ exists and,
hence, the flow $U$ is stable for
such large $k$.

Since  marginal stability occurs at $\lambda=0$ in \eqref{eigen},
we  analyze  the equation,
\begin{align}
\Ope{E}_I\hat{\phi}:=-\hat{\phi}''+k^2\hat{\phi}-\frac{U''}{U_I-U}\hat{\phi}=0.
\end{align}
If this solution is somehow available, we obtain the following stability criterion.
\begin{corollary}
 If $U(x)$ satisfies (A1) and (A2) and has a single inflection point
 $x_I$,  and  $\phi_c(x)$ denotes the  solution of
\begin{align}
\Ope{E}_I\phi_c=0,\quad \phi_c(-L)=0,\quad\phi_c'(-L)=1, 
\end{align}
then  \eqref{eq_w} is spectrally stable if and only if
$\phi_c(L)\ge0$.
\end{corollary}
\begin{proof}
According to Lemma~\ref{lem:Sturm}, 
$\phi_c$ does not have zero on $[-L,x_I]$ other than $x=-L$ and has at most one zero
 on $[x_I,L]$.
Note that, by increasing $\lambda$ from $0$, the general solution $\hat{\phi}$ of \eqref{eigen} becomes less
 oscillatory than $\phi_c$. 
 If $\phi_c(L)<0$, $\phi_c(x)$ has one zero on $[x_I,L]$ and hence there
 must be one eigenvalue $\lambda\in[0,1/k^2+\lambda_u]$ for which $\hat{\phi}$ satisfies \eqref{eigen} and \eqref{boundary}. 

Conversely, if $\phi_c(L)\ge0$, then $\phi_c(x)$ does not have zero on
 $-L<x<L$ and the solution $\hat{\phi}$ of \eqref{eigen} cannot satisfy
the boundary condition \eqref{boundary} when $\lambda>0$, i.e.,
there is no positive eigenvalue $\lambda>0$.
\end{proof}

In particular,  
we can obtain an analytical solution $\phi_c$ for the case of
$k\rightarrow0$ as
\begin{align}
\phi_c(x)
=&[U(-L)-U_I][U(x)-U_I]\int_{-L}^x\frac{1}{[U(s)-U_I]^2}ds.
\end{align}
Then, the necessary and sufficient stability condition $\phi_c(L)\ge0$
becomes
\begin{align}
\left.\frac{1}{U'(s)[U(s)-U_I]}\right|_{-L}^{L}
+\int_{-L}^{L}\frac{U''(s)}{U'^2(s)[U(s)-U_I]}ds\ge0,
\end{align}
which agrees with the result of Rosenbluth \& Simon~\cite{Rosenbluth}. (Note, the  typographical error in the final
equation (4) of this paper, in which $w'^3$ should be replaced by $w'^2$).

Another equivalent approach is to solve the equation $\Ope{E}_I\phi_c=0$ with
boundary conditions $\phi_c(-L)=\phi_c(L)=0$ and with a 
{\it derivative jump} at $x=x_I$, 
\begin{align}
 \phi_c(x_I+0)=\phi_c(x_I-0),\quad \alpha:=\phi_c'(x_I+0)-\phi_c'(x_I-0).
\end{align}
In other words, we solve $\Ope{E}_I\phi_c=-\alpha\delta(x-x_I)$ or
\begin{align}
 -\phi_c+\Ope{G}\left(\frac{U''}{U_I-U}\phi_c\right)=\alpha g(x,x_I).
\end{align}
By introducing a normalization
$\int_{-L}^{L}(-\phi_c''+k^2\phi_c)dx=1$ for $\phi_c$, we can determine
$\alpha$ as
\begin{align}
\alpha=-1+\int_{-L}^{L}\frac{U''}{U_I-U}\phi_c \, dx,
\end{align}
and arrive at the integral equation \eqref{BM} shown below.
This approach reproduces the  stability criterion obtained  by
Balmforth \& Morrison~\cite{Balmforth}: 

\begin{corollary}\label{th:BM}
 If $U(x)$ satisfies (A1) and (A2) and has a single inflection point
 $x_I$, and $\phi_c(x)$ denotes  the solution of
\begin{align}
-\phi_c(x)+g(x,x_I)+\int_{-L}^{L}[g(x,s)-g(x,x_I)]\frac{U''(s)}{U_I-U(s)}\phi_c(s)ds=0,
\label{BM}
\end{align}
then \eqref{eq_w} is spectrally stable if and only if
\begin{align}
 -1+\int_{-L}^{L}\frac{U''}{U_I-U}\phi_cdx<0,
\end{align}
\end{corollary}
\begin{proof}
 According to Lemma~\ref{lem:Sturm}, $\phi_c(x)$ does not have zero on
$-L<x<L$ and its sign should be always positive $\phi_c(x)>0$ due to the
 normalization.
If $\alpha>0$, we can eliminate this derivative jump by increasing
 $\lambda$ from $0$, since the general solution
 $\hat{\phi}$ of \eqref{eigen} becomes less
 oscillatory than $\phi_c$. Therefore,  there
 must be an eigenvalue $\lambda\in[0,1/k^2+\lambda_u]$ for which
 $\hat{\phi}$ satisfies \eqref{eigen} and \eqref{boundary} without the 
 derivative jump. 

Conversely, if $\alpha\le0$, this derivative jump gets large as
 $\lambda$ increases from $0$ and, hence,  there is no positive eigenvalue $\lambda>0$.
\end{proof}

\subsection{Multiple inflection points}

Here, we address the problem of multiple inflection points.
Recall from Proposition~\ref{th:discrete} that neutrally stable
eigenmodes may exist only at the frequencies $\omega=kU_{In}$, $n=1,2,\dots,N$.
In the same manner as for the case of a single inflection point, we consider the equations for the
neutrally stable eigenmodes,
\begin{align}
\Ope{E}_{In}\hat{\phi}_c:=-\hat{\phi}_c''+k^2\hat{\phi}_c-\frac{U''}{U_{In}-U}\hat{\phi}_c=0,\label{marginal}\\
 \hat{\phi}_c(-L)=\hat{\phi}_c(L)=0,
\end{align}
for every inflection point  $x_{In}$, $n=1,2,\dots,N$.
Since these equations do not have nontrivial solutions for general $k$,
we seek them for some characteristic values of $k$, in the same spirit as Tollmien's approach~\cite{Tollmien,Lin,Lin2,ZLin,ZLin2}.


\begin{proposition}\label{wavenumber}
Let $U(x)$ satisfy (A1) and (A2).
For each inflection point $x_{In}$, there is at most one critical
wavenumber $k_n>0$ at which the equation
\begin{align}
 \Ope{E}_{In}|_{k_n}\hat{\phi}_c=0, \quad\hat{\phi}_c(-L)=\hat{\phi}_c(L)=0,
\end{align}
has a nontrivial solution $\hat{\phi}_c$, where $\Ope{E}_{In}|_{k_n}$ denotes
the operator $\Ope{E}_{In}$ at $k=k_n$.
\end{proposition}
\begin{proof}
 According to Lemma~\ref{lem:Sturm}, the solution $\phi_c(x)$ of $\Ope{E}_{In}\phi_c=0$
 satisfying $\phi_c(-L)=0$ and $\phi_c'(-L)=1$ has at most one zero on $-L<x\le L$.
This $\phi_c(x)$ becomes less oscillatory as
 $k$ increases from $0$ to $\infty$ and eventually has no zero for
 $k^2>\max[U''/(U_{In}-U)]$. Therefore, there exists at
 most one value $k_n$ of $k$ for which $\phi_c(L)=0$ holds.
\end{proof}

Without loss of generality, let us focus on an inflection point $x_{I1}$ and
assume that there is a critical wavenumber $k_1>0$ for it. Namely, we
have a
solution $\hat{w}_c=-\hat{\phi}_c''+k_1^2\hat{\phi}_c\in\Spa{L}^2$ that satisfies
\begin{align}
 -\hat{w}_c+\frac{U''}{U_{I1}-U}\Ope{G}|_{k_1}\hat{w}_c=0,
\quad\mbox{or}\quad (\Ope{L}|_{k_1}-U_{I1})\hat{w}_c=0.\label{marginal2}
\end{align}
Now,
 we again invoke Corollary~\ref{th:extension} and consider the {\it
 selfadjoint} eigenvalue problem \eqref{eigen_multiple}.
The above neutrally stable eigenfunction $\hat{w}_c$ clearly corresponds to
the marginally stable eigenfunction ($\lambda=0$) of 
\eqref{eigen_multiple} at $k=k_1$, namely, $\Ope{H}_{\rm v}|_{k_1}\hat{w}_c=0$.

Let us continuously change the parameter $k$ in the neighborhood of
$k_1$ and investigate how an 
eigenvalue $\lambda$ and an eigenfunction $\hat{w}$ deviate
from $\lambda=0$ and $\hat{w}=\hat{w}_c$,
respectively.
By differentiating the identity,
\begin{align}
 0=&\int_{-L}^{L}\hat{w}(\lambda-\Ope{H}_{\rm v})\hat{w}dx,
\end{align}
with respect to $k$ and setting $k=k_1$, we obtain
\begin{align}
 0=&\int_{-L}^{L}\hat{w}_c\left(\left.\frac{\partial\lambda}{\partial
 k}\right|_{k_1}-\left.\frac{\partial\Ope{H}_{\rm v}}{\partial
 k}\right|_{k_1}\right)\hat{w}_cdx\nonumber\\
=&\int_{-L}^{L}\hat{w}_c
\left[\left.\frac{\partial\lambda}{\partial k}\right|_{k_1}-\frac{\nu}{U''}
\left.\frac{\partial\Ope{L}}{\partial k}\right|_{k_1}(U_{I1}-U_{I2})(U_{I1}-U_{I3})\dots(U_{I1}-U_{IN})\right]\hat{w}_c
 dx,
\end{align}
where \eqref{marginal2} has been used.
Since 
\begin{align}
 \frac{\partial\Ope{L}}{\partial k}=U''\frac{\partial\Ope{G}}{\partial k}=-2kU''\Ope{G}\Ope{G},
\end{align}
we get
\begin{align}
\left.\frac{\partial\lambda}{\partial
 k}\right|_{k_1}\|\hat{w}_c\|_{\Spa{L}^2}^2
=&-2k_1\nu(U_{I1}-U_{I2})(U_{I1}-U_{I3})\dots(U_{I1}-U_{IN})
\|\hat{\phi}_c\|_{\Spa{L}^2}^2.\label{lambda}
\end{align}
Similar relations are  available for the other critical wavenumbers $k_2,k_3,\dots,k_N$
if they exist.
In view of the condition \eqref{nu}, one can
 distinguish the sign
of $\partial\lambda/\partial k|_{k_n}$ from \eqref{lambda} as follows;
\begin{align}
 {\rm sgn}\left.\frac{\partial\lambda}{\partial
 k}\right|_{k_n}
={\rm sgn}[U'''(x_{In})U'(x_{In})]
={\rm sgn}[(U'^2)''(x_{In})],
\end{align}
which agrees with Tollmien and Lin's result~\cite{Tollmien,Lin,Lin2}. In other words, if the absolute value of the
background vorticity $|U'(x)|$ has a local maximum (or minimum) at
$x=x_{In}$, then a positive eigenvalue $\lambda>0$ emerges at $k=k_n$ as $k$
decreases (or increases).

We note that there is no positive eigenvalue $\lambda>0$ of
\eqref{eigen_multiple} in the
limit of $k\rightarrow\infty$. As
$k$ continuously changes from $\infty$ to $0$,
the number of positive eigenvalues increases (or decreases) by one when
$k$ passes through $k_n$ that is associated with the inflection point
$x_{In}$ satisfying $(U'^2)''(x_{In})<0$ (or $>0$).
We can summarize these facts into the
following stability criterion.

\begin{corollary}\label{th:index}
 Let $U(x)$ satisfy (A1) and (A2). Suppose that, for every inflection points $x_{In}$,
 $n=1,2,\dots,N$, the critical wavenumbers $k_n>0$, $n=1,2,\dots,N$,
 are either solved or proven to be nonexistent  according to
 Proposition~\ref{wavenumber}. Then, the equation \eqref{eq_w} is spectrally unstable 
if and only if $N^+-N^->0$, where

$N^+$: number of the critical wavenumbers $k_n$ that satisfy $k<k_n$ and
$(U'^2)''(x_{In})<0$,

$N^-$: number of the critical wavenumbers $k_n$ that satisfy $k<k_n$ and $(U'^2)''(x_{In})>0$.
\end{corollary}

When $N^+-N^-$ is positive, it corresponds to the number of
positive eigenvalues $\lambda$ of \eqref{eigen_multiple}. 
This number cannot be greater than the number of the inflection
points $x_{In}$ satisfying $(U'^2)''(x_{In})<0$, i.e., the number
of local maxima of $|U'(x)|$.

A similar result to Corollary~\ref{th:index} is shown by
Lin~\cite{ZLin,ZLin2} as a rigorous justification of Tollmien's method.
While he treats a larger class of flows than ours,
his criterion is sufficient but not necessary for
instability in the presence of multiple
inflection points~\cite{ZLin2}.    Balmforth \& Morrison~\cite{Balmforth} have also discussed the case of multiple
inflection points in the same manner as Corollary~\ref{th:BM}, where the
derivative jump $\alpha_n$ is evaluated for each inflection point
$x_{In}$ and then $\alpha_n<0$ (or $\alpha_n>0$) corresponds to $k<k_n$ (or $k>k_n$). However, in this work the   importance of ${\rm sgn}(U'^2)''(x_{In})$ was not observed.

\section{Numerical tests}\label{sec:numerical}

Finally, we exhibit  numerical results to illustrate the practicability of our 
method. For three velocity profiles $U(x)$, we 
compare the results of two different numerical codes:  one code solves
the Rayleigh equation~\eqref{eq_phi} directly for  complex
eigenvalues $c=\omega/k\in\Bod{C}$, while the other code solves for 
the eigenvalues $\lambda_1,\lambda_2,\dots, $ of the selfadjoint
operator $\Ope{H}_{\rm v}$ in descending order.

The first example is
\begin{align}
 U(x)=\tanh(x),\quad x\in[-\infty,\infty],
\end{align}
which is well-known to be unstable for $0<k<1$. The result is shown in
Fig.~\ref{fig:tanh}, where we
also plot $\tilde{\lambda}_1=\max Q/\|\xi\|_{\Spa{L}^2}^2$ for
comparison (the damping eigenvalue ${\rm Im}\, c<0$ is not plotted since its presence is trivial).
As  expected from the results of \ref{sec:extension} in Sec.~\ref{sec:proof}, $\lambda=0$ is the upper
edge of the continuous spectrum of $\Ope{H}$. 
Since the eigenfunction $\hat{\xi}_1$ becomes singular, i.e., 
$\|\hat{\xi}_1\|_{\Spa{L}^2}\rightarrow\infty$,  as
$\tilde{\lambda}_1\rightarrow+0$,
the curve of $\tilde{\lambda}_1$ is tangent to the marginal line
$\lambda=0$ and the critical
wavenumber $k=1$ is not so evident. On the other
hand, the upper
edge of the continuous spectrum of $\Ope{H}_{\rm v}$ is less than zero,
$\lambda_u=\max[\tanh(x)/\tanh''(x)]=-0.5<0$, and hence the maximum
eigenvalue $\lambda_1$ of $\Ope{H}_{\rm v}$ smoothly intersect with
$\lambda=0$ at $k=1$ in Fig.~\ref{fig:tanh}. Thus, for the purpose of drawing the stability boundary,
the variational principle with respect to the norm $\|w\|_{\Spa{L}^2}$ is seen to be numerically efficient and accurate.

\begin{figure}
 \centering\includegraphics[width=8cm]{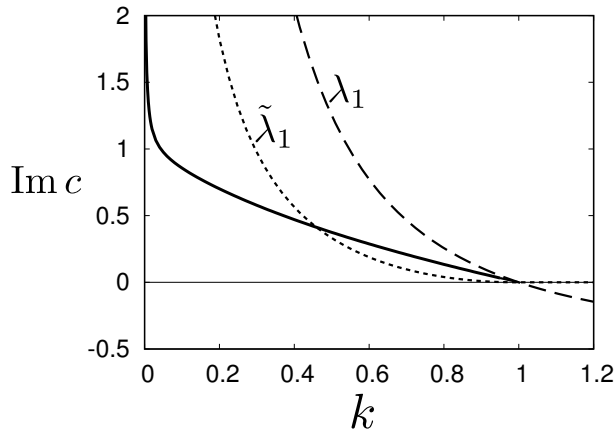}
 \caption{Growth rate ${\rm Im}\, c$ (where ${\rm Re}\, c\equiv0$),
 $\lambda_1=\max Q/\|w\|_{\Spa{L}^2}^2$ and $\tilde{\lambda}_1=\max
 Q/\|\xi\|_{\Spa{L}^2}^2$  versus wavenumber $k$ for the shear flow $U(x)=\tanh(x)$.}
\label{fig:tanh}
\end{figure}

The second example is
\begin{align}
 U(x)=x + 5x^3 + 1.62 \tanh[4(x - 0.5)],\quad x\in[-1,1],
\end{align}
which was previously  addressed by Balmforth \& Morrison~\cite{Balmforth}. This flow
has three inflection points,
\begin{align}
\begin{split}
 x_{I1} &= -0.069, \quad U_{I1} = -1.65,\\
 x_{I2} &= 0.622,\quad \ \   U_{I2} = 2.55,\\
 x_{I3} &= 0.665,\quad  \ \ U_{I3} = 3.07, 
\end{split}
\end{align}
at which $(U'^2)''$ is positive, negative,  and positive, respectively.
Only for $x_{I2}$ and $x_{I3}$, do  the critical wavenumbers
$k_2\simeq1.2$ and $k_3\simeq0.4$ exist. As predicted in
Corollary~\ref{th:index}, the instability occurs only for finite wavenumbers
$k_3<k<k_2$. In Fig.~\ref{fig:finite_k}, the positive signature of the
maximum eigenvalue $\lambda_1$ certainly
agrees with this unstable regime. In practice, our variational approach can directly prove the instability
at a fixed $k$ without knowing the existence of nor the  values  $k_1,k_2$ and $k_3$.

\begin{figure}
\centering\includegraphics[width=8cm]{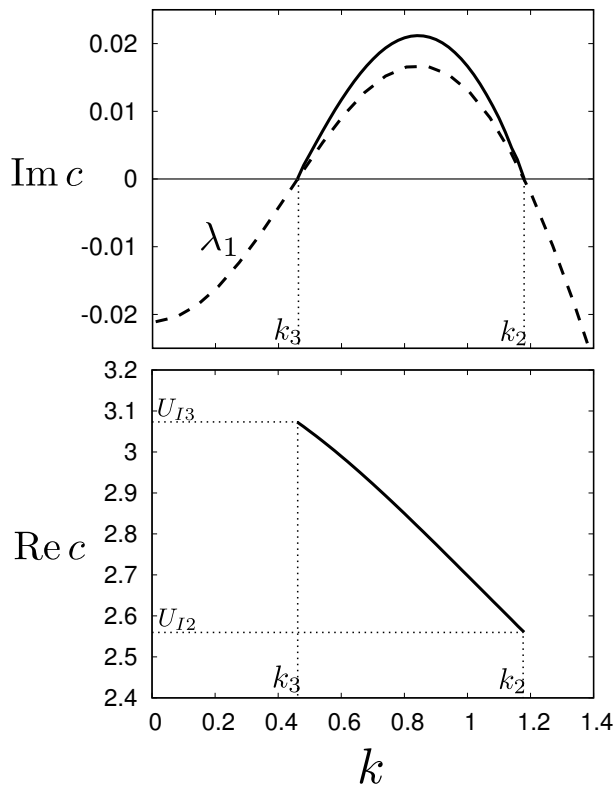}
 \caption{Growth rate (${\rm Im}\, c$)
 and phase speed (${\rm Re}\, c$) versus wavenumber $k$
 for the shear flow $U(x)=x + 5x^3 + 1.62 \tanh[4(x - 0.5)]$.
The dashed line is $\lambda_1=\max Q/\|w\|_{\Spa{L}^2}^2$.
}
\label{fig:finite_k}
\end{figure}

The third example is
\begin{align}
 U(x)=x-0.02+\sin[8(x-0.02)]/16,\quad x\in[-1,1],
\end{align}
which has five inflection points,
\begin{align}
\begin{split}
 x_{I1} &= -0.765, \quad U_{I1} = -0.785,\\
 x_{I2} &= -0.373,\quad U_{I2} = -0.393,\\
 x_{I3} &= 0.020,\quad \  \ U_{I3} = 0.0,\\
 x_{I4} &= 0.413,\quad \ \ U_{I4} = 0.393,\\
 x_{I5} &= 0.805,\quad \ \ U_{I5} = 0.785. 
\end{split}
\end{align}
For this example there exist three critical wavenumbers
$k_1$, $k_3$, and $k_5$
for the inflection points $x_{I1},x_{I3},x_{I5}$,  all of  which have 
$(U'^2)''$ negative. Therefore, three unstable eigenvalues emerge
at $k_1,k_3,k_5$ with different phase speeds $U_{I1},U_{I3},U_{I5}$,
respectively. Thus, three eigenvalues $\lambda_1,\lambda_2,\lambda_3$ of our
variational problem
completely predict the onsets of instabilities, as shown in Fig.~\ref{fig:multi}.

\begin{figure}
\centering\includegraphics[width=8cm]{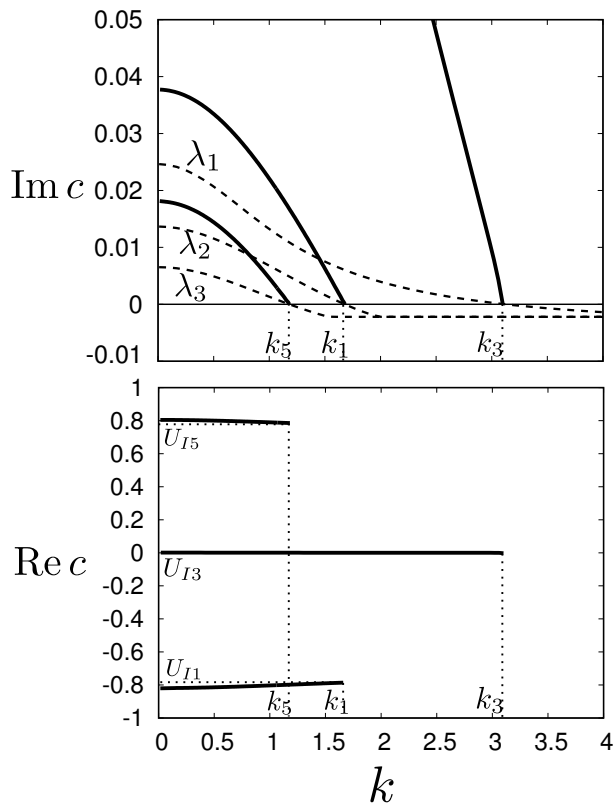}
 \caption{Growth rate (${\rm Im}\, c$)
 and phase speed (${\rm Re}\, c$) versus wavenumber $k$ 
 for the shear flow $U(x)=x-0.02+\sin[8(x-0.02)]/16$.
The dashed lines are eigenvalues $\lambda_1,\lambda_2,\lambda_3$ of
 $\Ope{H}_{\rm v}$.
}
\label{fig:multi}
\end{figure}

\section{Summary}
\label{summary}

We have investigated the linear stability of inviscid plane parallel shear flow
(Rayleigh's equation) as a typical example of an infinite-dimensional and
non-selfadjoint eigenvalue problem that originates upon linearizing a 
Hamiltonian system. By assuming monotonicity and analyticity of the shear
profile, a necessary and sufficient condition
for spectral stability was obtained  in the form of a variational criterion
(Theorem~\ref{th:multiple}), in which a positive signature of the
quadratic form $Q=\langle\xi,\Ope{H}\xi\rangle$ 
 implies  existence of an unstable eigenmode. Since $\Ope{H}$ is selfadjoint, we were able to prove
instability must occur if  some test function $\xi$ (virtual displacement) exists that makes $Q$ positive,
which is analytically and numerically easier to do than solving Rayleigh's equation.
Moreover, the singularity at the stability boundary (due to the continuous spectrum)
was shown to be  removed technically by maximizing $Q$ with respect to the vorticity disturbance
$w\in\Spa{L}^2$,  instead of the displacement $\xi\in\Spa{L}^2$. However,
we remark that, unlike the Rayleigh-Ritz method, neither 
$\max Q/\|\xi\|_{\Spa{L}^2}^2$ nor $\max Q/\|w\|_{\Spa{L}^2}^2$ are  quantitatively related to
the maximum growth rate of instability.

Our variational criterion is  an improvement of previous sufficient stability criteria  \cite{Arnold,Barston}.
Given that Rayleigh's equation has been solved under a specific condition,
we have also reproduced the earlier results of the Nyquist method~\cite{Rosenbluth,Balmforth}
and Tollmien's analysis of the neutral modes~\cite{Tollmien,Lin,Lin2}.


In this paper, we have imposed the assumptions (A1) and (A2) on 
velocity profile $U(x)$ to simplify the discussion. The relaxation of
these assumptions is possible to some extent, but it would be difficult to
overcome the following difficulties:  (i)  If  analyticity is not assumed and  $U''$ is only continuous,  
special care is needed for piecewise-linear regions of
$U$. In such a  region, say $[x_1,x_2]$, all points are regarded as 
inflection points  and we expect that the variational problem would become 
a minimax problem like $\min_{x_I\in[x_1,x_2]}\max_{\xi\in\Spa{L}^2}Q>0$, which is not so
analytically tractable. (ii)  If  monotonicity is not assumed, a serious difficulty arises when
the sign of $U''$ is not identical at the locations of multiple critical
layers for a phase speed $c=\omega/k\in\Bod{R}$. Since at  this frequency
$\omega=kc$ belongs to degenerate multiple continuous spectra whose signature is indefinite, our technique for
constructing $Q$  breaks down.

In conclusion we note that our variational approach will be applicable to rather
simple equilibrium profiles which are free from the above difficulties.  However, there is a large class of  fluid and plasma systems with existing sufficient stability criteria (e.g., magnetohydrodynamics \cite{Frieman,ampIIa} with flow)  that have  
Kre\u{i}n-like  signature (or  action-angle variables) for a  continuous spectrum.  This is the key ingredient  needed for  constructing the  quadratic form. Thus,  our techniques are available for  a large class of applications governed by  other dynamical systems.  We will report our additional results  in future publications.

\section*{Acknowledgment}

The authors would like to thank Z. Yoshida,  Y. Fukumoto, and G. Hagstrom 
for fruitful discussions.   This work was supported by a grant-in-aid for scientific research from the Japan  Society for the Promotion of Science (No. 25800308).
P.J.M.\   was  supported by U.S. Dept.\ of Energy Contract \# DE-FG05-80ET-53088.

\end{document}